\documentclass{article}%
\usepackage{amssymb}
\usepackage{graphicx}
\usepackage{amsmath}
\usepackage{amsfonts}%
\providecommand{\U}[1]{\protect\rule{.1in}{.1in}}
\newtheorem{theorem}{Theorem}

\newtheorem{corollary}[theorem]{Corollary}

\newtheorem{definition}[theorem]{Definition}

\newtheorem{lemma}[theorem]{Lemma}
\newtheorem{notation}[theorem]{Notation}

\newtheorem{proposition}[theorem]{Proposition}

\newenvironment{proof}[1][Proof]{\textbf{#1.} }{\ \rule{0.5em}{0.5em}}
\begin{document}

\title{From the Biot-Savart Law\\to Amp\`{e}re's Magnetic Circuital Law\\via Synthetic Differential Geometry}
\author{Hirokazu NISHIMURA\\Institute of Mathematics, University of Tsukuba\\Tsukuba, Ibaraki, 305-8571\\JAPAN}
\maketitle

\begin{abstract}
It is well known in classical electrodynamics that the magnetic field given by
a current loop and the electric field caused by the corresponding electric
dipoles in sheets are very similar, as far as we are far away from the loop,
which enables us to deduce Amp\`{e}re's magnetic circuital law from the
Biot-Savart law easily. The principal objective in this paper is to show that
synthetic differential geometry, in which nilpotent infinitesimals are
available in abundance, furnishes out a natural framework for the exquisite
formulation of this similitude and its demonstration. This similitude in
heaven enables us to transit from the Biot-Savart law to Amp\`{e}re's magnetic
circuital law like a shot on earth.

\end{abstract}

\section{\label{s1}Introduction}

It is well known among physicists (see, e.g., \cite{wa}) that the magnetic
field given by a \textit{current loop} and the electric field caused by the
corresponding \textit{electric dipoles} in sheets are very similar, as far as
we are far away from the loop, which enables us to deduce \textit{Amp\`{e}re's
magnetic circuital law} from the \textit{Biot-Savart law} easily. However, a
mathematically satisfactory formulation of this similitude is by no means
easy, let alone its proof based upon the Coulomb and Biot-Savart laws.

In good old days of the 17th and 18th centuries, mathematicians and physicists
could communicate easily with ones of the other species, and many excellent
mathematicians were physicists at the same time and vice versa. The honeymoon
was over when mathematicians rushed into eradication of their shabby
\textit{nilpotent infinitesimals} by replacing them with their authentic
$\varepsilon$-$\delta$ arguments.

In the middle of the 20th century, moribund nilpotent infinitesimals were
resurrected in not earthly but heavenly manners by \textit{synthetic
differential geometers}. They have constructed another world of mathematics,
called a \textit{well-adapted model} (a kind of Grothendieck toposes), in
which they could indulge themselves in their favorite nilpotent
infinitesimals. We have a route from the earth to heaven
(\textit{internalization}) and another route in the opposite direction
(\textit{externalization}), so that our synthetic formulation and
demonstration of the similitude is of earthly significance. For synthetic
differential geometry, the reader is referred to \cite{ko} and \cite{la}.

The very similitude is formulated and established synthetically in
\S \ref{s4}, which is preceded by a synthetical approach to electric dipoles
in sheets in \S \ref{s3}. Once the similitude is firmly established within a
well-adapted model, some of its consequences are externalized, which enables
us to derive the Amp\`{e}re's magnetic circuital law from the Biot-Savart law,
as is seen in \S \ref{s5}. In a subsequent paper, we will discuss Vassiliev
invariants in knot theory (cf. \cite{va1} and \cite{va2}) from this standpoint.

\section{\label{s2}Preliminaries}

In this section we fix our notation for static electric fields and static
magnetic fields. Since we would like to concentrate upon mathematical aspects,
we omit unnecessary physical constants or the like from this standpoint.

\subsection{\label{s2.1}Static Electric Fields}

Given a figure $\Omega$\ in $\mathbf{R}^{3}$\ and a mapping $q:\Omega
\rightarrow\mathbf{R}$ (as the density of electric charge), the static
electric field $\mathbf{E}_{\left(  \Omega,q\right)  }:\mathbf{R}%
^{3}\rightarrow\mathbf{R}^{3}$ associated with $\left(  \Omega,q\right)  $\ is
given by an integral. Namely, the Coulomb law tells us that
\[
\mathbf{E}_{\left(  \Omega,q\right)  }\left(  \mathbf{x}\right)  =\int
_{\Omega}\frac{q\left(  \mathbf{p}\right)  \left(  \mathbf{x}-\mathbf{p}%
\right)  }{\left\|  \mathbf{x}-\mathbf{p}\right\|  ^{3}}dp
\]
for any $\mathbf{x}\in\mathbf{R}^{3}$, where the integral is the volume
integral, the surface integral or the line integral according to whether the
figure $\Omega$\ is three-dimensional, two-dimensional or one-dimensional. As
is well known, the following Maxwell equations obtain:%

\begin{align}
\mathrm{div}\,\mathbf{E}_{\left(  \Omega,q\right)  }  &  =4\pi q\label{2.1.1}%
\\
\mathrm{rot}\,\mathbf{E}_{\left(  \Omega,q\right)  }  &  =\mathbf{0}%
\label{2.1.2}%
\end{align}
Now we consider electric dipoles. Let $S$\ be an oriented surface in
$\mathbf{R}^{3}$\ and $\sigma,h\in\mathbf{R}$. Let $\mathbf{n}_{S}%
:S\rightarrow\mathbf{R}^{3}$ be the unit normal in the positive direction. We
slide the surface $S$\ by $\frac{h}{2}\mathbf{n}_{S}$ to get the surface
$S_{\frac{h}{2}}$. The surface $S_{\frac{h}{2}}$ endowed with the constant
density $\sigma$\ of electric charge gives rise to the static electric field
$\mathbf{E}_{\left(  S,\sigma,h\right)  }^{+}$\ by the Coulomb law. Similarly,
We slide the surface $S$\ by $-\frac{h}{2}\mathbf{n}_{S}$ to get the surface
$S_{-\frac{h}{2}}$. The surface $S_{-\frac{h}{2}}$ endowed with the constant
density $-\sigma$\ of electric charge gives rise to the static electric field
$\mathbf{E}_{\left(  S,\sigma,h\right)  }^{-}$\ by the Coulomb law. They
together yield the static electric field
\[
\mathbf{E}_{\left(  S,\sigma,h\right)  }^{\mathrm{dp}}=\mathbf{E}_{\left(
S,\sigma,h\right)  }^{+}+\mathbf{E}_{\left(  S,\sigma,h\right)  }^{-}%
\]
by the Coulomb law.

\subsection{\label{s2.2}The Biot-Savart Law and Amp\`{e}re's Magnetic
Circuital Law}

The static magnetic field caused by a current loop is given by the so-called
Biot-Savart law, so that, given a loop $C:t\in\left[  0,t_{0}\right]
\mapsto\mathbf{m}\left(  t\right)  \in\mathbf{R}^{3}$, it gives rise to its
static magnetic field $\mathbf{B}_{C}$ by%

\begin{align}
\mathbf{B}_{C}\left(  \mathbf{x}\right)   &  =\frac{1}{4\pi}\int_{0}^{t_{0}%
}\frac{\left(  \mathbf{x}-\mathbf{m}\left(  t\right)  \right)  \times
\frac{d\mathbf{m}}{dt}\left(  t\right)  }{\left\|  \mathbf{x}-\mathbf{m}%
\left(  t\right)  \right\|  ^{3}}dt\nonumber\\
&  =\frac{1}{4\pi}\int_{C}\frac{\left(  \mathbf{x}-\mathbf{r}\right)  \times
d\mathbf{r}}{\left\|  \mathbf{x}-\mathbf{r}\right\|  ^{3}}\label{2.2.1}%
\end{align}
for any $\mathbf{x}\in\mathbf{R}^{3}$, where $\mathbf{r}$\ moves along the
curve $C$. Given another loop $L:s\in\left[  0,s_{0}\right]  \mapsto
\mathbf{l}\left(  t\right)  \in\mathbf{R}^{3}$, Amp\`{e}re's magnetic
circuital law claims that
\begin{align}
&  \frac{1}{4\pi}\int_{0}^{s_{0}}ds\int_{0}^{t_{0}}dt\frac{\left(  \left(
\mathbf{l}\left(  s\right)  -\mathbf{m}\left(  t\right)  \right)  \times
\frac{d\mathbf{m}}{dt}\left(  t\right)  \right)  \cdot\frac{d\mathbf{l}}%
{ds}\left(  s\right)  }{\shortparallel\mathbf{l}\left(  s\right)
-\mathbf{m}\left(  t\right)  \shortparallel^{3}}\nonumber\\
&  =\frac{1}{4\pi}\int_{L}\int_{C}\frac{\left(  \left(  \mathbf{s}%
-\mathbf{r}\right)  \times d\mathbf{r}\right)  \cdot d\mathbf{s}}{\left\|
\mathbf{s}-\mathbf{r}\right\|  ^{3}}\nonumber\\
&  =\mathbf{Lk}\left(  C,L\right) \label{2.2.2}%
\end{align}
where $\mathbf{s}$\ moves along the curve $L$, and $\mathbf{Lk}\left(
C,L\right)  $ is defined as follows:

\begin{definition}
Let $S$ be an oriented surface with its induced oriented boundary $L$, which
is supposed to be transversal to $C$\ at their intersecting points. They are
enumerated as
\[
S\cap C=\left\{  \mathbf{p}_{1},...,\mathbf{p}_{k}\right\}  \text{.}%
\]
We define $\varepsilon_{i}$ ($i=1,...,k$) to be $1$ if the tangent of $C$\ at
$\mathbf{p}_{i}$ transits $S$\ into the part that the orientation of
$S$\ selects, and $-1$ otherwise. Now we define
\[
\mathbf{Lk}\left(  C,L\right)  =\sum_{i=1}^{n}\varepsilon_{i}%
\]
The reader should note that the definition is independent of our choice of $S
$.
\end{definition}

Topology tells us that

\begin{proposition}
The number $\mathbf{Lk}\left(  C,L\right)  $ has the following properties:

\begin{enumerate}
\item It is symmetric in the sense that
\[
\mathbf{Lk}\left(  C,L\right)  =\mathbf{Lk}\left(  L,C\right)
\]

\item For any oriented surface $S$ with $\partial S=L\cup-L^{\prime}$, if it
does not intersect $C$, then we have
\[
\mathbf{Lk}\left(  C,L\right)  =\mathbf{Lk}\left(  C,L^{\prime}\right)
\]
where $-L^{\prime}$ denotes the same curve $L^{\prime}$ with the orientation reversed.
\end{enumerate}
\end{proposition}

\begin{notation}
The first and the second formulas of (\ref{2.2.2}) is denoted by
\[
\mathbf{A}\left(  C,L\right)
\]

\end{notation}

\section{\label{s3}Synthetic Differential Geometry of Electric Dipoles in
Infinitesimal Sheets}

In this and the subsequent sections we are working within a well-adapted model.

\begin{notation}
We denote by $\mathbb{R}$\ the set of real numbers containing nilpotent
infinitesimals in abundance (called a line object in synthetic differential
geometry). We denote by $\mathbb{R}_{+}$\ the set
\[
\left\{  x\in\mathbb{R\mid}x>0\right\}
\]
We denote by $D$\ the set
\[
\left\{  d\in\mathbb{R}\mid d^{2}=0\right\}
\]
Intuitively, $D$ stands for the set of first-order infinitesimals.
\end{notation}

Let $m$ be an integer and $n$ a natural number. For the mapping
\[
x\in\mathbb{R}_{+}\mathbb{\mapsto}x^{m}\in\mathbb{R}%
\]
we have
\begin{equation}
\left(  x+d\right)  ^{m}=x^{m}+mx^{m-1}d\label{3.1}%
\end{equation}
for any $d\in D$,\ as is well known. For the mapping
\[
x\in\mathbb{R}_{+}\mathbb{\mapsto}x^{\frac{m}{n}}\in\mathbb{R}%
\]
we have

\begin{lemma}
\label{t3.1}
\[
\left(  x+d\right)  ^{\frac{m}{n}}=x^{\frac{m}{n}}+\frac{m}{n}x^{\frac{m}%
{n}-1}d
\]

\end{lemma}

\begin{proof}
By the Kock-Lawvere axiom, there exists a unique\ $a\in\mathbb{R}$ such that
\[
\left(  x+d\right)  ^{\frac{m}{n}}=x^{\frac{m}{n}}+ad
\]
for any $d\in D$. On the one hand, we have
\[
\left(  \left(  x+d\right)  ^{\frac{m}{n}}\right)  ^{n}=\left(  x+d\right)
^{m}=x^{m}+mx^{m-1}d
\]
by (\ref{3.1}). On the other hand, we have
\[
\left(  x^{\frac{m}{n}}+ad\right)  ^{n}=\left(  x^{\frac{m}{n}}\right)
^{n}+n\left(  x^{\frac{m}{n}}\right)  ^{n-1}ad
\]
by the binomial theorem. Therefore we have
\[
mx^{m-1}=n\left(  x^{\frac{m}{n}}\right)  ^{n-1}a
\]
so that
\[
a=\frac{m}{n}x^{m-1}\left(  x^{\frac{m}{n}}\right)  ^{1-n}=\frac{m}{n}%
x^{\frac{n\left(  m-1\right)  +m\left(  1-n\right)  }{n}}=\frac{m}{n}%
x^{\frac{m}{n}-1}%
\]

\end{proof}

\begin{corollary}
\label{t3.1.1}Let $\mathbf{x},\mathbf{a}\in\mathbb{R}^{3}$ with $\mathbf{x}%
\neq\mathbf{0}$. Then we have
\[
\left\|  \mathbf{x}+\mathbf{a}d\right\|  ^{-3}=\left\|  \mathbf{x}\right\|
^{-3}-3\left\|  \mathbf{x}\right\|  ^{-5}\left(  \mathbf{x}\cdot
\mathbf{a}\right)  d
\]
for any $d\in D$, where $\left\|  \mathbf{x}\right\|  $ is the standard norm
of $\mathbf{x}$\ (i.e., $\left\|  \mathbf{x}\right\|  =\sqrt{\left(
x_{1}\right)  ^{2}+\left(  x_{2}\right)  ^{2}+\left(  x_{3}\right)  ^{2}}$
with $\mathbf{x}=\left(  x_{1},x_{2},x_{3}\right)  $) and $\cdot$ stands for
the inner product.
\end{corollary}

\begin{proof}
We have
\begin{align*}
&  \left\|  \mathbf{x}+\mathbf{a}d\right\|  ^{-3}\\
&  =\left(  \left\|  \mathbf{x}+\mathbf{a}d\right\|  ^{2}\right)  ^{-\frac
{3}{2}}\\
&  =\left(  \left(  \mathbf{x}+\mathbf{a}d\right)  \cdot\left(  \mathbf{x}%
+\mathbf{a}d\right)  \right)  ^{-\frac{3}{2}}\\
&  =\left(  \left(  \mathbf{x}\cdot\mathbf{x}\right)  +2\left(  \mathbf{x}%
\cdot\mathbf{a}\right)  d\right)  ^{-\frac{3}{2}}\\
&  =\left\|  \mathbf{x}\right\|  ^{-3}-3\left\|  \mathbf{x}\right\|
^{-5}\left(  \mathbf{x}\cdot\mathbf{a}\right)  d\\
&  \text{\lbrack By Lemma \ref{t3.1}]}%
\end{align*}

\end{proof}

\begin{proposition}
\label{t3.2}Let $d,e,h\in D$, $\sigma\in\mathbb{R}$ and $\mathbf{x}%
,\mathbf{a},\mathbf{b},\mathbf{r}\in\mathbb{R}^{3}$ with $\mathbf{x}%
\neq\mathbf{r}$, $\mathbf{a}\times\mathbf{b}\neq\mathbf{0}$. Let $S$\ be the
infinitesimal parallelogram spanned by $\mathbf{x}$, $\mathbf{x}+d\mathbf{a}%
$\ and $\mathbf{x}+e\mathbf{b}$.
\[%
\begin{array}
[c]{ccc}%
\mathbf{x} &
\begin{array}
[c]{c}%
e\mathbf{b}\\
\rightarrow
\end{array}
& \mathbf{x}+e\mathbf{b}\\%
\begin{array}
[c]{cc}%
d\mathbf{a} & \downarrow
\end{array}
& S &
\begin{array}
[c]{cc}%
\downarrow & d\mathbf{a}%
\end{array}
\\
\mathbf{x}+d\mathbf{a} &
\begin{array}
[c]{c}%
\rightarrow\\
e\mathbf{b}%
\end{array}
& \mathbf{x}+d\mathbf{a}+e\mathbf{b}%
\end{array}
\]
Then we have
\[
\mathbf{E}_{\left(  S,\sigma,h\right)  }^{\mathrm{dp}}\left(  \mathbf{r}%
\right)  =\frac{h\sigma de}{\left\|  \mathbf{r}-\mathbf{x}\right\|  ^{3}%
}\left(  3\left(  \frac{\mathbf{r}-\mathbf{x}}{\left\|  \mathbf{r}%
-\mathbf{x}\right\|  }\cdot\left(  \mathbf{a}\times\mathbf{b}\right)  \right)
\frac{\mathbf{r}-\mathbf{x}}{\left\|  \mathbf{r}-\mathbf{x}\right\|  }-\left(
\mathbf{a}\times\mathbf{b}\right)  \right)
\]

\end{proposition}

\begin{proof}%
\begin{align*}
\mathbf{E}_{\left(  S,\sigma,h\right)  }^{+}\left(  \mathbf{r}\right)   &
=\sigma de\left\|  \mathbf{a}\times\mathbf{b}\right\|  \left\|  \mathbf{r}%
-\left(  \mathbf{x}+\frac{h}{2}\frac{\mathbf{a}\times\mathbf{b}}{\left\|
\mathbf{a}\times\mathbf{b}\right\|  }\right)  \right\|  ^{^{-3}}\left(
\mathbf{r}-\left(  \mathbf{x}+\frac{h}{2}\frac{\mathbf{a}\times\mathbf{b}%
}{\left\|  \mathbf{a}\times\mathbf{b}\right\|  }\right)  \right) \\
&  =\sigma de\left\|  \mathbf{a}\times\mathbf{b}\right\|  \left\|  \left(
\mathbf{r}-\mathbf{x}\right)  -\frac{h}{2}\frac{\mathbf{a}\times\mathbf{b}%
}{\left\|  \mathbf{a}\times\mathbf{b}\right\|  }\right\|  ^{-3}\left(  \left(
\mathbf{r}-\mathbf{x}\right)  -\frac{h}{2}\frac{\mathbf{a}\times\mathbf{b}%
}{\left\|  \mathbf{a}\times\mathbf{b}\right\|  }\right) \\
&  =\sigma de\left\|  \mathbf{a}\times\mathbf{b}\right\|  \left(  \left\|
\mathbf{r}-\mathbf{x}\right\|  ^{-3}+\frac{3h\left(  \left(  \mathbf{r}%
-\mathbf{x}\right)  \cdot\left(  \mathbf{a}\times\mathbf{b}\right)  \right)
}{2\left\|  \mathbf{a}\times\mathbf{b}\right\|  }\left\|  \mathbf{r}%
-\mathbf{x}\right\|  ^{-5}\right) \\
&  \left(  \left(  \mathbf{r}-\mathbf{x}\right)  -\frac{h}{2}\frac
{\mathbf{a}\times\mathbf{b}}{\left\|  \mathbf{a}\times\mathbf{b}\right\|
}\right) \\
&  \text{[By Corollary \ref{t3.1.1}]}%
\end{align*}
On the other hand, we have
\begin{align*}
\mathbf{E}_{\left(  S,\sigma,h\right)  }^{-}\left(  \mathbf{r}\right)   &
=-\sigma de\left\|  \mathbf{a}\times\mathbf{b}\right\|  \left\|
\mathbf{r}-\left(  \mathbf{x}-\frac{h}{2}\frac{\mathbf{a}\times\mathbf{b}%
}{\left\|  \mathbf{a}\times\mathbf{b}\right\|  }\right)  \right\|  ^{^{-3}%
}\left(  \mathbf{r}-\left(  \mathbf{x}-\frac{h}{2}\frac{\mathbf{a}%
\times\mathbf{b}}{\left\|  \mathbf{a}\times\mathbf{b}\right\|  }\right)
\right) \\
&  =-\sigma de\left\|  \mathbf{a}\times\mathbf{b}\right\|  \left\|  \left(
\mathbf{r}-\mathbf{x}\right)  +\frac{h}{2}\frac{\mathbf{a}\times\mathbf{b}%
}{\left\|  \mathbf{a}\times\mathbf{b}\right\|  }\right\|  ^{-3}\left(  \left(
\mathbf{r}-\mathbf{x}\right)  +\frac{h}{2}\frac{\mathbf{a}\times\mathbf{b}%
}{\left\|  \mathbf{a}\times\mathbf{b}\right\|  }\right) \\
&  =-\sigma de\left\|  \mathbf{a}\times\mathbf{b}\right\|  \left(  \left\|
\mathbf{r}-\mathbf{x}\right\|  ^{-3}-\frac{3h\left(  \left(  \mathbf{r}%
-\mathbf{x}\right)  \cdot\left(  \mathbf{a}\times\mathbf{b}\right)  \right)
}{2\left\|  \mathbf{a}\times\mathbf{b}\right\|  }\left\|  \mathbf{r}%
-\mathbf{x}\right\|  ^{-5}\right) \\
&  \left(  \left(  \mathbf{r}-\mathbf{x}\right)  +\frac{h}{2}\frac
{\mathbf{a}\times\mathbf{b}}{\left\|  \mathbf{a}\times\mathbf{b}\right\|
}\right) \\
&  \text{[By Corollary \ref{t3.1.1}]}%
\end{align*}
Therefore we have
\begin{align*}
\mathbf{E}_{\left(  S,\sigma,h\right)  }^{\mathrm{bp}}\left(  \mathbf{r}%
\right)   &  =\mathbf{E}_{\left(  S,\sigma,h\right)  }^{+}\left(
\mathbf{r}\right)  +\mathbf{E}_{\left(  S,\sigma,h\right)  }^{-}\left(
\mathbf{r}\right) \\
&  =\sigma de\left\|  \mathbf{a}\times\mathbf{b}\right\|  \left(
\begin{array}
[c]{c}%
\frac{3h\left(  \left(  \mathbf{r}-\mathbf{x}\right)  \cdot\left(
\mathbf{a}\times\mathbf{b}\right)  \right)  \left\|  \mathbf{r}-\mathbf{x}%
\right\|  ^{-5}}{\left\|  \mathbf{a}\times\mathbf{b}\right\|  }\left(
\mathbf{r}-\mathbf{x}\right)  -\\
\frac{h\left\|  \mathbf{r}-\mathbf{x}\right\|  ^{-3}}{\left\|  \mathbf{a}%
\times\mathbf{b}\right\|  }\left(  \mathbf{a}\times\mathbf{b}\right)
\end{array}
\right) \\
&  =\frac{h\sigma de}{\left\|  \mathbf{r}-\mathbf{x}\right\|  ^{3}}\left(
3\left(  \frac{\mathbf{r}-\mathbf{x}}{\left\|  \mathbf{r}-\mathbf{x}\right\|
}\cdot\left(  \mathbf{a}\times\mathbf{b}\right)  \right)  \frac{\mathbf{r}%
-\mathbf{x}}{\left\|  \mathbf{r}-\mathbf{x}\right\|  }-\left(  \mathbf{a}%
\times\mathbf{b}\right)  \right)
\end{align*}

\end{proof}

\section{\label{s4}The Similitude between the Electric Fields of Electric
Dipoles in Sheets and the Magnetic Fields of Current Loops within Synthetic
Differential Geometry}

The principal objective in this section is to establish the similitude between
the electric fields of dipoles in sheets and the magnetic fields of current
loops synthetically. The discussion is very similar to that in Stokes'
theorem, for which the reader is referred to \cite{ni1}, \cite{ni2} and
\cite{ni3}. Let us begin with

\begin{lemma}
\label{t4.1}For any vectors $\mathbf{a},\mathbf{b}\in\mathbb{R}^{3}$ and any
unit vector $\widehat{\mathbf{r}}\in\mathbb{R}^{3}$, we have
\[
\left(  \left(  \mathbf{a}\times\mathbf{b}\right)  \cdot\widehat{\mathbf{r}%
}\right)  \widehat{\mathbf{r}}=\mathbf{a}\times\mathbf{b+}\left(
\widehat{\mathbf{r}}\cdot\mathbf{a}\right)  \mathbf{b\times}\widehat
{\mathbf{r}}-\left(  \widehat{\mathbf{r}}\cdot\mathbf{b}\right)
\mathbf{a\times}\widehat{\mathbf{r}}%
\]

\end{lemma}

\begin{proof}
Fixing arbitrarily $\widehat{\mathbf{r}}=\left(  \widehat{r}_{1},\widehat
{r}_{2},\widehat{r}_{3}\right)  $ with $\left(  \widehat{r}_{1}\right)
^{2}+\left(  \widehat{r}_{2}\right)  ^{2}+\left(  \widehat{r}_{3}\right)
^{2}=1$, both the left-hand and the right-hand of the above formula can be
regarded as functions of $\left(  \mathbf{a},\mathbf{b}\right)  =\left(
\left(  \widehat{a}_{1},\widehat{a}_{2},\widehat{a}_{3}\right)  ,\left(
\widehat{b}_{1},\widehat{b}_{2},\widehat{b}_{3}\right)  \right)  \in
\mathbb{R}^{3}\times\mathbb{R}^{3}$. It is easy to see that both functions are
bilinear, so that it suffices to show the above formula in cases of
$\mathbf{a}=\mathbf{i},\mathbf{j},\mathbf{k}$ and $\mathbf{b}=\mathbf{i}%
,\mathbf{j},\mathbf{k}$, where $\mathbf{i},\mathbf{j},\mathbf{k}$ are the
standard base of $\mathbb{R}^{3}$, namely, $\mathbf{i}=\left(  1,0,0\right)
$, $\mathbf{j}=\left(  0,1,0\right)  $ and $\mathbf{k}=\left(  0,0,1\right)
$. In case of $\mathbf{a}=\mathbf{b}$, it is easy to see that both sides
degenerate into $\mathbf{0}$. In case of $\mathbf{a}=\mathbf{i}$ and
$\mathbf{b}=\mathbf{j}$, we have $\mathbf{a}\times\mathbf{b}=\mathbf{k}$, so
that the left-hand is $\left(  \widehat{r}_{3}\widehat{r}_{1},\widehat{r}%
_{3}\widehat{r}_{2},\left(  \widehat{r}_{3}\right)  ^{2}\right)  $, while the
right-hand is
\begin{align*}
&  \left(  0,0,1\right)  +\widehat{r}_{1}\left(  \mathbf{j}\times
\widehat{\mathbf{r}}\right)  -\widehat{r}_{2}\left(  \mathbf{i}\times
\widehat{\mathbf{r}}\right) \\
&  =\left(  0,0,1\right)  +\left(  \widehat{r}_{1}\widehat{r}_{3},0,-\left(
\widehat{r}_{1}\right)  ^{2}\right)  -\left(  0,-\widehat{r}_{2}\widehat
{r}_{3},\left(  \widehat{r}_{2}\right)  ^{2}\right) \\
&  =\left(  \widehat{r}_{3}\widehat{r}_{1},\widehat{r}_{3}\widehat{r}%
_{2},\left(  \widehat{r}_{3}\right)  ^{2}\right) \\
&  \text{[since }\left(  \widehat{r}_{1}\right)  ^{2}+\left(  \widehat{r}%
_{2}\right)  ^{2}+\left(  \widehat{r}_{3}\right)  ^{2}\text{ is equal to
}1\text{]}%
\end{align*}
The remaining five cases are safely left to the reader.
\end{proof}

\begin{theorem}
\label{t4.2}(The Infinitesimal Similitude) Let $d,e\in D$ and $\mathbf{x}%
,\mathbf{a},\mathbf{b}\in\mathbb{R}^{3}$ with $\mathbf{a}\times\mathbf{b}%
\neq\mathbf{0}$. Let $C$ be the infinitesimal oriented curve moving from
$\mathbf{x}$\ to $\mathbf{x}+d\mathbf{a}$ by $d\mathbf{a}$, moving from
$\mathbf{x}+d\mathbf{a}$ to $\mathbf{x}+d\mathbf{a}+e\mathbf{b}$\ by
$e\mathbf{b}$, moving from $\mathbf{x}+d\mathbf{a}+e\mathbf{b}$\ to
$\mathbf{x}+e\mathbf{b}$\ by $-d\mathbf{a}$\ and finally moving from
$\mathbf{x}+e\mathbf{b}$\ to the start $\mathbf{x}$\ by $-e\mathbf{b}$.
\[%
\begin{array}
[c]{ccc}%
\mathbf{x} &
\begin{array}
[c]{c}%
-e\mathbf{b}\\
\longleftarrow
\end{array}
& \mathbf{x}+e\mathbf{b}\\%
\begin{array}
[c]{cc}%
d\mathbf{a} & \downarrow
\end{array}
& S &
\begin{array}
[c]{cc}%
\uparrow & -d\mathbf{a}%
\end{array}
\\
\mathbf{x}+d\mathbf{a} &
\begin{array}
[c]{c}%
\rightarrow\\
e\mathbf{b}%
\end{array}
& \mathbf{x}+d\mathbf{a}+e\mathbf{b}%
\end{array}
\]
Let $S$\ be the infinitesimal oriented parallelogram spanned by $\mathbf{x}$,
$\mathbf{x}+d\mathbf{a}$\ and $\mathbf{x}+e\mathbf{b}$\ with its induced
oriented boundary $C$. Let $h\in D$ and $\sigma\in\mathbb{R}$. Then we have
\[
\mathbf{E}_{\left(  S,1,h\right)  }^{\mathrm{dp}}\left(  \mathbf{r}\right)
=h\mathbf{B}_{C}\left(  \mathbf{r}\right)
\]
for any $\mathbf{r}\in\mathbb{R}^{3}$ with $\mathbf{r}\neq\mathbf{x}$.
\end{theorem}

\begin{proof}
On the one hand, thanks to Proposition \ref{t3.2}, we have
\[
\mathbf{E}_{\left(  S,1,h\right)  }^{\mathrm{dp}}\left(  \mathbf{r}\right)
=\frac{hde}{\left\Vert \mathbf{r}-\mathbf{x}\right\Vert ^{3}}\left(  3\left(
\frac{\mathbf{r}-\mathbf{x}}{\left\Vert \mathbf{r}-\mathbf{x}\right\Vert
}\cdot\left(  \mathbf{a}\times\mathbf{b}\right)  \right)  \frac{\mathbf{r}%
-\mathbf{x}}{\left\Vert \mathbf{r}-\mathbf{x}\right\Vert }-\left(
\mathbf{a}\times\mathbf{b}\right)  \right)
\]
On the other hand, we have
\begin{align*}
\mathbf{B}_{C}\left(  \mathbf{r}\right)   &  =\frac{d\mathbf{a}\times\left(
\mathbf{r}-\mathbf{x}\right)  }{\left\Vert \mathbf{r}-\mathbf{x}\right\Vert
^{3}}+\frac{e\mathbf{b}\times\left(  \mathbf{r}-\left(  \mathbf{x}%
+d\mathbf{a}\right)  \right)  }{\left\Vert \mathbf{r}-\left(  \mathbf{x}%
+d\mathbf{a}\right)  \right\Vert ^{3}}-\frac{d\mathbf{a}\times\left(
\mathbf{r}-\left(  \mathbf{x}+e\mathbf{b}\right)  \right)  }{\left\Vert
\mathbf{r}-\left(  \mathbf{x}+e\mathbf{b}\right)  \right\Vert ^{3}}%
-\frac{e\mathbf{b}\times\left(  \mathbf{r}-\mathbf{x}\right)  }{\left\Vert
\mathbf{r}-\mathbf{x}\right\Vert ^{3}}\\
&  =\frac{d\mathbf{a}\times\left(  \mathbf{r}-\mathbf{x}\right)  }{\left\Vert
\mathbf{r}-\mathbf{x}\right\Vert ^{3}}+\frac{e\mathbf{b}\times\left(  \left(
\mathbf{r}-\mathbf{x}\right)  -d\mathbf{a}\right)  }{\left\Vert \left(
\mathbf{r}-\mathbf{x}\right)  -d\mathbf{a}\right\Vert ^{3}}-\frac
{d\mathbf{a}\times\left(  \left(  \mathbf{r}-\mathbf{x}\right)  -e\mathbf{b}%
\right)  }{\left\Vert \left(  \mathbf{r}-\mathbf{x}\right)  -e\mathbf{b}%
\right\Vert ^{3}}-\frac{e\mathbf{b}\times\left(  \mathbf{r}-\mathbf{x}\right)
}{\left\Vert \mathbf{r}-\mathbf{x}\right\Vert ^{3}}\\
&  =\frac{d\mathbf{a}\times\left(  \mathbf{r}-\mathbf{x}\right)  }{\left\Vert
\mathbf{r}-\mathbf{x}\right\Vert ^{3}}+\left(  e\mathbf{b}\times\left(
\left(  \mathbf{r}-\mathbf{x}\right)  -d\mathbf{a}\right)  \right)  \left(
\left\Vert \mathbf{r}-\mathbf{x}\right\Vert ^{-3}+3\left\Vert \mathbf{r}%
-\mathbf{x}\right\Vert ^{-5}\left(  \left(  \mathbf{r}-\mathbf{x}\right)
\cdot\mathbf{a}\right)  d\right)  -\\
&  \left(  d\mathbf{a}\times\left(  \left(  \mathbf{r}-\mathbf{x}\right)
-e\mathbf{b}\right)  \right)  \left(  \left\Vert \mathbf{r}-\mathbf{x}%
\right\Vert ^{-3}+3\left\Vert \mathbf{r}-\mathbf{x}\right\Vert ^{-5}\left(
\left(  \mathbf{r}-\mathbf{x}\right)  \cdot\mathbf{b}\right)  e\right)
-\frac{e\mathbf{b}\times\left(  \mathbf{r}-\mathbf{x}\right)  }{\left\Vert
\mathbf{r}-\mathbf{x}\right\Vert ^{3}}\\
&  =\left\{  \left(  e\mathbf{b}\times\left(  \left(  \mathbf{r}%
-\mathbf{x}\right)  -d\mathbf{a}\right)  \right)  \left(  \left\Vert
\mathbf{r}-\mathbf{x}\right\Vert ^{-3}+3\left\Vert \mathbf{r}-\mathbf{x}%
\right\Vert ^{-5}\left(  \left(  \mathbf{r}-\mathbf{x}\right)  \cdot
\mathbf{a}\right)  d\right)  -\frac{e\mathbf{b}\times\left(  \mathbf{r}%
-\mathbf{x}\right)  }{\left\Vert \mathbf{r}-\mathbf{x}\right\Vert ^{3}%
}\right\}  -\\
&  \left\{  \left(  d\mathbf{a}\times\left(  \left(  \mathbf{r}-\mathbf{x}%
\right)  -e\mathbf{b}\right)  \right)  \left(  \left\Vert \mathbf{r}%
-\mathbf{x}\right\Vert ^{-3}+3\left\Vert \mathbf{r}-\mathbf{x}\right\Vert
^{-5}\left(  \left(  \mathbf{r}-\mathbf{x}\right)  \cdot\mathbf{b}\right)
e\right)  -\frac{d\mathbf{a}\times\left(  \mathbf{r}-\mathbf{x}\right)
}{\left\Vert \mathbf{r}-\mathbf{x}\right\Vert ^{3}}\right\}  \\
&  =\left\{  -de\left\Vert \mathbf{r}-\mathbf{x}\right\Vert ^{-3}\left(
\mathbf{b}\times\mathbf{a}\right)  +3de\left\Vert \mathbf{r}-\mathbf{x}%
\right\Vert ^{-5}\left(  \left(  \mathbf{r}-\mathbf{x}\right)  \cdot
\mathbf{a}\right)  \left(  \mathbf{b}\times\left(  \mathbf{r}-\mathbf{x}%
\right)  \right)  \right\}  -\\
&  \left\{  -de\left\Vert \mathbf{r}-\mathbf{x}\right\Vert ^{-3}\left(
\mathbf{a}\times\mathbf{b}\right)  +3de\left\Vert \mathbf{r}-\mathbf{x}%
\right\Vert ^{-5}\left(  \left(  \mathbf{r}-\mathbf{x}\right)  \cdot
\mathbf{b}\right)  \left(  \mathbf{a}\times\left(  \mathbf{r}-\mathbf{x}%
\right)  \right)  \right\}  \\
&  =de\left\Vert \mathbf{r}-\mathbf{x}\right\Vert ^{-3}\left\{
\begin{array}
[c]{c}%
2\left(  \mathbf{a}\times\mathbf{b}\right)  +\\
3\left(  \frac{\mathbf{r}-\mathbf{x}}{\left\Vert \mathbf{r}-\mathbf{x}%
\right\Vert }\cdot\mathbf{a}\right)  \left(  \mathbf{b}\times\frac
{\mathbf{r}-\mathbf{x}}{\left\Vert \mathbf{r}-\mathbf{x}\right\Vert }\right)
-3\left(  \frac{\mathbf{r}-\mathbf{x}}{\left\Vert \mathbf{r}-\mathbf{x}%
\right\Vert }\cdot\mathbf{b}\right)  \left(  \mathbf{a}\times\frac
{\mathbf{r}-\mathbf{x}}{\left\Vert \mathbf{r}-\mathbf{x}\right\Vert }\right)
\end{array}
\right\}
\end{align*}
Therefore the desired result follows by dint of Lemma \ref{t4.1}.
\end{proof}

\begin{theorem}
\label{t4.3}(The General Similitude) Let $S$ be an oriented surface with its
induced oriented boundary $C$. Let $h\in D$. Then we have
\[
\mathbf{E}_{\left(  S,1,h\right)  }^{\mathrm{dp}}\left(  \mathbf{r}\right)
=h\mathbf{B}_{C}\left(  \mathbf{r}\right)
\]
for any $\mathbf{r}\in\mathbb{R}^{3}$ with $\mathbf{r}\notin S$.
\end{theorem}

\begin{proof}
We divide the oriented surface $S$ into $MN$ infinitesimal oriented
parallelograms, where $M$ and $N$ are very great natural numbers. It is
depicted partially and schematically in the following diagram:
\[%
\begin{array}
[c]{ccccc}%
\mathbf{x}_{i,j} & \leftarrow & \mathbf{x}_{i,j+1} & \leftarrow &
\mathbf{x}_{i,j+2}\\
\downarrow & S_{i,j} &
\begin{array}
[c]{cc}%
\uparrow & \downarrow
\end{array}
& S_{i,j+1} & \uparrow\\
\mathbf{x}_{i+1,j} &
\begin{array}
[c]{c}%
\rightarrow\\
\leftarrow
\end{array}
& \mathbf{x}_{i+1,j+1} &
\begin{array}
[c]{c}%
\rightarrow\\
\leftarrow
\end{array}
& \mathbf{x}_{i+1,j+2}\\
\downarrow & S_{i+1,j} &
\begin{array}
[c]{cc}%
\uparrow & \downarrow
\end{array}
& S_{i+1,j+1} & \uparrow\\
\mathbf{x}_{i+2,j} & \rightarrow & \mathbf{x}_{i+2,j+1} & \rightarrow &
\mathbf{x}_{i+2,j+2}%
\end{array}
\]
Then surely we have
\begin{equation}
\mathbf{E}_{\left(  S,1,h\right)  }^{\mathrm{dp}}\left(  \mathbf{r}\right)
=\sum_{i=0}^{M-1}\sum_{j=0}^{N-1}\mathbf{E}_{\left(  S_{i,j},1,h\right)
}^{\mathrm{dp}}\left(  \mathbf{r}\right) \label{4.3.1}%
\end{equation}
Proposition \ref{t3.2} enables us to conclude that
\begin{equation}
\sum_{i=0}^{M-1}\sum_{j=0}^{N-1}\mathbf{E}_{\left(  S_{i,j},1,h\right)
}^{\mathrm{dp}}\left(  \mathbf{r}\right)  =h\sum_{i=0}^{M-1}\sum_{j=0}%
^{N-1}\mathbf{B}_{C_{i,j}}\left(  \mathbf{r}\right) \label{4.3.2}%
\end{equation}
The boundary $C_{i,j}$ of the infinitesimal parallelogram $S_{i,j}$ consists
of the infinitesimal segment from $\mathbf{x}_{i,j}$ to $\mathbf{x}_{i+1,j}$,
that from $\mathbf{x}_{i+1,j}$\ to $\mathbf{x}_{i+1,j+1}$, that from
$\mathbf{x}_{i+1,j+1}$ to $\mathbf{x}_{i,j+1}$ and that from $\mathbf{x}%
_{i,j+1}$ to $\mathbf{x}_{i,j}$. Unless $i=M-1$, the second segment from
$\mathbf{x}_{i+1,j}$\ to $\mathbf{x}_{i+1,j+1}$ is shared by the infinitesimal
parallelogram $S_{i+1,j}$ as its boundary in the opposite direction.
Similarly, unless $j=N-1$, the third segment from $\mathbf{x}_{i+1,j+1}$ to
$\mathbf{x}_{i,j+1}$ is shared by the infinitesimal parallelogram $S_{i,j+1}$
as its boundary in the opposite direction. Therefore we have
\begin{equation}
\sum_{i=0}^{M-1}\sum_{j=0}^{N-1}\mathbf{B}_{C_{i,j}}\left(  \mathbf{r}\right)
=\mathbf{B}_{C}\left(  \mathbf{r}\right) \label{4.3.3}%
\end{equation}
Therefore the desired formula follows readily from (\ref{4.3.1}),
(\ref{4.3.2}) and (\ref{4.3.3}).
\end{proof}

\begin{corollary}
\label{t4.3.1}With the same notation and assumptions in the above theorem, we
have
\[
\left(  \mathrm{rot}\,\mathbf{B}_{C}\right)  \left(  \mathbf{r}\right)
=\mathbf{0}%
\]

\end{corollary}

\begin{proof}
We have
\begin{align*}
&  h\left(  \mathrm{rot}\,\mathbf{B}_{C}\right)  \left(  \mathbf{r}\right) \\
&  =\left(  \mathrm{rot}\,\mathbf{E}_{\left(  S,1,h\right)  }^{\mathrm{dp}%
}\right)  \left(  \mathbf{r}\right) \\
&  \text{\lbrack By Theorem \ref{t4.3}]}\\
&  =\mathbf{0}\\
&  \text{\lbrack By (\ref{2.1.2})]}%
\end{align*}
for any $h\in D$, so that we have
\[
\left(  \mathrm{rot}\,\mathbf{B}_{C}\right)  \left(  \mathbf{r}\right)
=\mathbf{0}%
\]

\end{proof}

\section{\label{s5}From the Biot-Savart Law to Amp\`{e}re's Circuital Law}

This section owes much to \cite{fu}.

\begin{proposition}
\label{t5.1}For any $\mathbf{x}\notin C$, we have
\[
\left(  \mathrm{rot}\,\mathbf{B}_{C}\right)  \left(  \mathbf{x}\right)
=\mathbf{0}%
\]

\end{proposition}

\begin{proof}
Since $\mathbf{x}\notin C$, it is not difficult to find a surface $S$\ dodging
$\mathbf{x}$ with its boundary being $C$. By internalizing these entities in a
well-adapted model and externalizing Corollary \ref{t4.3.1}, we get the
desired result.
\end{proof}

\begin{proposition}
\label{t5.2}The number $\mathbf{A}\left(  C,L\right)  $ has the following properties:

\begin{enumerate}
\item It is symmetric in the sense that
\[
\mathbf{A}\left(  C,L\right)  =\mathbf{A}\left(  L,C\right)
\]

\item For any oriented surface $S$ with $\partial S=L\cup-L^{\prime}$, if it
does not intersect $C$, then we have
\[
\mathbf{A}\left(  C,L\right)  =\mathbf{A}\left(  C,L^{\prime}\right)
\]

\end{enumerate}
\end{proposition}

\begin{proof}
The first property follows simply from
\begin{align*}
&  \left(  \left(  \mathbf{l}\left(  s\right)  -\mathbf{m}\left(  t\right)
\right)  \times\frac{d\mathbf{m}}{dt}\left(  t\right)  \right)  \cdot
\frac{d\mathbf{l}}{ds}\left(  s\right) \\
&  =\det\,\left(
\begin{array}
[c]{c}%
\mathbf{l}\left(  s\right)  -\mathbf{m}\left(  t\right) \\
\frac{d\mathbf{m}}{dt}\left(  t\right) \\
\frac{d\mathbf{l}}{ds}\left(  s\right)
\end{array}
\right) \\
&  =\det\,\left(
\begin{array}
[c]{c}%
\mathbf{m}\left(  t\right)  -\mathbf{l}\left(  s\right) \\
\frac{d\mathbf{l}}{ds}\left(  s\right) \\
\frac{d\mathbf{m}}{dt}\left(  t\right)
\end{array}
\right) \\
&  =\left(  \left(  \mathbf{m}\left(  t\right)  -\mathbf{l}\left(  s\right)
\right)  \times\frac{d\mathbf{l}}{ds}\left(  s\right)  \right)  \cdot
\frac{d\mathbf{m}}{dt}\left(  t\right)
\end{align*}
The second property follows simply from Stokes' theorem, as is seen in the
following computation:
\begin{align*}
&  \mathbf{A}\left(  C,L\right)  -\mathbf{A}\left(  C,L^{\prime}\right) \\
&  =\frac{1}{4\pi}\int_{L\cup-L^{\prime}}\mathbf{B}_{L}\cdot d\mathbf{r}\\
&  =\frac{1}{4\pi}\int_{S}\left(  \mathrm{rot}\,\mathbf{B}_{C}\right)  \cdot
d\mathbf{S}\\
&  \text{[By Stokes' Theorem]}\\
&  =0\\
&  \text{[By Proposition \ref{t5.1}]}%
\end{align*}

\end{proof}

\begin{lemma}
\label{t5.3}Let $n$ be a natural number with $n\geq2$. The curve $L$ is the
unit circle on the $xy$ plane with center $\left(  0,0,0\right)  $ rounding
counterclockwise against the positive part of the $z$ axis. The curve $C_{n}$,
to begin with, moves up straight from $\left(  0,0,-n\right)  $ to $\left(
0,0,n\right)  $, moves horizontally from $\left(  0,0,n\right)  $ to $\left(
n,0,n\right)  $, moves down straight from $\left(  n,0,n\right)  $ to $\left(
n,0,-n\right)  $, and finally moves horizontally from $\left(  n,0,-n\right)
$ to $\left(  0,0,-n\right)  $.
\[%
\begin{array}
[c]{ccc}%
\left(  0,0,n\right)  & \rightarrow & \left(  n,0,n\right) \\
\uparrow &  & \downarrow\\
\circlearrowleft &  & \downarrow\\
\uparrow &  & \downarrow\\
\left(  0,0,-n\right)  & \leftarrow & \left(  n,0,-n\right)
\end{array}
\]
Then we have
\[
\mathbf{A}\left(  C_{n},L\right)  =1
\]
while trivially we have
\[
\mathbf{Lk}\left(  C_{n},L\right)  =1
\]

\end{lemma}

\begin{proof}
Thanks to Proposition \ref{t5.2}, we are sure that $\mathbf{A}\left(
C_{n},L\right)  $ is independent of $n$, for we have
\[
\mathbf{A}\left(  C_{n}\cup-C_{n+1},L\right)  =0
\]
as is to be seen easily. The curve $C_{n}$\ is composed of the curve
$C_{n}^{1}$\ moving up straight from $\left(  0,0,-n\right)  $ to $\left(
0,0,n\right)  $ and the curve $C_{n}^{2}$ moving horizontally from $\left(
0,0,n\right)  $ to $\left(  n,0,n\right)  $, then moving down straight from
$\left(  n,0,n\right)  $ to $\left(  n,0,-n\right)  $ and finally moves
horizontally from $\left(  n,0,-n\right)  $ to $\left(  0,0,-n\right)  $. Now
we have
\[
\mathbf{A}\left(  C_{n},L\right)  =\frac{1}{4\pi}\int_{L}\int_{C_{n}^{1}}%
\frac{\left(  \left(  \mathbf{s}-\mathbf{r}\right)  \times d\mathbf{r}\right)
\cdot d\mathbf{s}}{\left\Vert \mathbf{s}-\mathbf{r}\right\Vert ^{3}}+\frac
{1}{4\pi}\int_{L}\int_{C_{n}^{2}}\frac{\left(  \left(  \mathbf{s}%
-\mathbf{r}\right)  \times d\mathbf{r}\right)  \cdot d\mathbf{s}}{\left\Vert
\mathbf{s}-\mathbf{r}\right\Vert ^{3}}%
\]
where $\mathbf{s}$ moves along the curve $L$\ and $\mathbf{r}\ $moves along
the curve $C_{n}^{1}$\ or $C_{n}^{2}$. It is easy to see that we have
\[
\frac{1}{4\pi}\int_{L}\int_{C_{n}^{2}}\frac{\left(  \left(  \mathbf{s}%
-\mathbf{r}\right)  \times d\mathbf{r}\right)  \cdot d\mathbf{s}}{\left\Vert
\mathbf{s}-\mathbf{r}\right\Vert ^{3}}\rightarrow0
\]
as $n\rightarrow\infty$, while we have
\[
\frac{1}{4\pi}\int_{L}\int_{C_{n}^{1}}\frac{\left(  \left(  \mathbf{s}%
-\mathbf{r}\right)  \times d\mathbf{r}\right)  \cdot d\mathbf{s}}{\left\Vert
\mathbf{s}-\mathbf{r}\right\Vert ^{3}}\rightarrow\frac{1}{4\pi}\int_{L}%
\int_{C^{\infty}}\frac{\left(  \left(  \mathbf{s}-\mathbf{r}\right)  \times
d\mathbf{r}\right)  \cdot d\mathbf{s}}{\left\Vert \mathbf{s}-\mathbf{r}%
\right\Vert ^{3}}%
\]
as $n\rightarrow\infty$, where the curve $C^{\infty}$\ is no other than the
$z$-axis moving from $-\infty$\ to $+\infty$. It is well known that
\[
\frac{1}{4\pi}\int_{L}\int_{C^{\infty}}\frac{\left(  \left(  \mathbf{s}%
-\mathbf{r}\right)  \times d\mathbf{r}\right)  \cdot d\mathbf{s}}{\left\Vert
\mathbf{s}-\mathbf{r}\right\Vert ^{3}}=1
\]
Therefore we are done.
\end{proof}

\begin{theorem}
\label{t5.4}(The General Amp\`{e}re's Circuital Law) The Amp\`{e}re's law
(\ref{2.2.2}) obtains.
\end{theorem}

\begin{proof}
Let $\varepsilon$\ be a very small positive number. To each $t\in\left[
0,t_{0}\right]  $, we consider the circle $\mathcal{C}_{\varepsilon}\left(
t\right)  $ with its center $\mathbf{m}\left(  t\right)  $ and its radius
$\varepsilon$ in the plane perpendicular to $\frac{d\mathbf{m}}{dt}\left(
t\right)  $. Then the totality of $\mathcal{C}\left(  t\right)  $ with $t$
ranging over $\left[  0,T\right]  $ forms a cylinder-like figure, which cuts
out $k$ circle-like curves from $S$. They are denoted by $L_{1},...,L_{k}$,
which surround the surfaces $S_{1},...,S_{k}$\textbf{\ }containing
$\mathbf{p}_{1},...,\mathbf{p}_{k}$, respectively. They are endowed with the
orientations induced from that of the surface $S$. Then the surface
$S^{\prime}$\ carved out by the curve $L\cup\left(  -L_{1}\right)  \cup
...\cup\left(  -L_{k}\right)  $ from $S$ no longer intersects the curve $C$,
so that we have
\[
\mathbf{A}\left(  C,L\cup\left(  -L_{1}\right)  \cup...\cup\left(
-L_{k}\right)  \right)  =0
\]
by dint of Stokes' Theorem and Proposition \ref{t5.1}. On the other hand, we
are sure by the very definition that
\[
\mathbf{A}\left(  C,L\cup\left(  -L_{1}\right)  \cup...\cup\left(
-L_{k}\right)  \right)  =\mathbf{A}\left(  C,L\right)  -\sum_{i=1}%
^{k}\mathbf{A}\left(  C,L_{i}\right)
\]
while we have%

\[
\mathbf{A}\left(  C,L_{i}\right)  =\mathbf{Lk}\left(  C,L_{i}\right)
\]
by dint of Lemma \ref{t5.3} with the aid of Proposition \ref{t5.2}. Therefore
we are done.
\end{proof}

\end{document}